\documentclass[10pt,english]{article}
\usepackage[T1]{fontenc}
\usepackage[latin9]{inputenc}
\usepackage[letterpaper]{geometry}
\usepackage{color}
\usepackage{units}
\usepackage{amsthm}
\usepackage{amsmath}
\usepackage{amssymb}
\usepackage{setspace}
\makeatletter
\newcommand{\lyxaddress}[1]{
\par {\raggedright #1
\vspace{1.4em}
\noindent\par}
}
  \theoremstyle{definition}
  \newtheorem{defn}{\protect\definitionname}
  \theoremstyle{plain}
  \newtheorem{prop}{\protect\propositionname}
  \theoremstyle{plain}
  \newtheorem{lem}{\protect\lemmaname}

\@ifundefined{date}{}{\date{}}
\usepackage[numbers,sort&compress]{natbib}

\usepackage{multirow}

\def\frontmatter@abstractheading{}

\date{}

\makeatother

\usepackage{babel}
  \providecommand{\definitionname}{Definition}
  \providecommand{\lemmaname}{Lemma}
  \providecommand{\propositionname}{Proposition}

\begin{document}

\title{Quantum Models for Psychological Measurements: \\
An Unsolved Problem}

\author{Andrei Khrennikov\textsuperscript{1}, Irina Basieva\textsuperscript{1},
Ehtibar N. Dzhafarov\textsuperscript{2}%
\thanks{Corresponding author: Department of Psychological Sciences, Purdue
University, West Lafayette, IN, USA; ehtibar@purdue.edu%
}, and \textcolor{black}{Jerome R. Busemeyer}\textsuperscript{3}}

\maketitle

\lyxaddress{\begin{center}
\textsuperscript{1}Linnaeus University, Sweden, \textsuperscript{2}Purdue
University, USA, \textsuperscript{3}\textcolor{black}{Indiana University,
USA}
\par\end{center}}
\begin{abstract}
There has been a strong recent interest in applying quantum mechanics
(QM) outside physics, including in cognitive science. We analyze the
applicability of QM to two basic properties in opinion polling. The
first property (response replicability) is that, for a large class
of questions, a response to a given question is expected to be repeated
if the question is posed again, irrespective of whether another question
is asked and answered in between. The second property (question order
effect) is that the response probabilities frequently depend on the
order in which the questions are asked. Whenever these two properties
occur together, it poses a problem for QM. The conventional QM with
Hermitian operators can handle response replicability, but only in
the way incompatible with the question order effect. In the generalization
of QM known as theory of positive-operator-valued measures (POVMs),
in order to account for response replicability, the POVMs involved
must be conventional operators. Although these problems are not unique
to QM and also challenge conventional cognitive theories, they stand
out as important unresolved problems for the application of QM to
cognition . Either some new principles are needed to determine the
bounds of applicability of QM to cognition, or quantum formalisms
more general than POVMs are needed.

\end{abstract}

\section{Introduction}

Quantum mechanics (QM) is the mathematical formalism of quantum physics.
(Sometimes the two are considered synonymous, in which case what we
call here QM would have to be called ``mathematical formalism of
QM.'') However, QM has recently begun to be used in various domains
outside of physics, e.g., in biology and economics \cite{AS2}--\cite{POLINA},
as well as in cognitive science \cite{K2}--\cite{Dzhafarov4}.
See recently published monographs \cite{Jerome5} and \cite{H2} for
overviews, as well as the recent target article in \emph{Brain and
Behavioral Sciences} \cite{POTZ} with ensuing commentaries and rejoinders.
There is one obvious similarity between cognitive science and quantum
physics: both deal with observations that are fundamentally probabilistic.
This similarity makes the use of QM in cognitive science plausible,
as QM is specifically designed to deal with random variables. Here,
we analyze the applicability of QM in opinion-polling, and compare
it to psychophysical judgments. 

On a very general level, QM accounts for the probability distributions
of measurement results using two kinds of entities, called \emph{observables}
$A$ and \emph{states} $\psi$ (of the system on which the measurements
are made). Let us assume that measurements are performed in a series
of consecutive trials numbered $1,2,\ldots$. In each trial $t$ the
experimenter decides what measurement to make (e.g., what question
to ask), and this amounts to choosing an observable $A$. Despite
its name, the latter is not observable per se, in the colloquial sense
of the word, but it is associated with a certain set of values $v\left(A\right)$,
which are the possible results one can observe by measuring $A$.
In a psychological experiment these are the responses that a participant
is allowed to give, such as \emph{Yes} and \emph{No}.

The probabilities of these outcomes in trial $t$ (conditioned on
all the previous measurements and their outcomes) are computed as
some function of the observable $A$ and of the state $\psi^{\left(t\right)}$
in which the system (a particle in quantum physics, or a participant
in psychology) is at the beginning of trial $t$, 
\begin{equation}
\Pr\left[v\left(A\right)=v\textnormal{ in trial }t\,\vert\,\mbox{measurements in trials }1,\ldots,t-1\right]=F\left(\psi^{\left(t\right)},A,v\right).\label{eq: very general v}
\end{equation}
This measurement changes the state of the system, so that at the end
of trial $t$ the state is $\psi^{\left(t+1\right)}$, generally different
from $\psi^{\left(t\right)}$. The change $\psi^{\left(t\right)}\rightarrow\psi^{\left(t+1\right)}$
depends on the observable $A$, the state $\psi^{\left(t\right)}$,
and the value $v\left(A\right)$ observed in trial $t$, 
\begin{equation}
\psi^{\left(t+1\right)}=G\left(\psi^{\left(t\right)},A,v\right).\label{eq: very general rho}
\end{equation}

On this level of generality, a psychologist will easily recognize
in (\ref{eq: very general v})-(\ref{eq: very general rho}) a probabilistic
version of the time-honored Stimulus-Organism-Response (S-O-R) scheme
for explaining behavior \cite{WOO}. This scheme involves stimuli
(corresponding to $A$), responses (corresponding to $v$), and internal
states (corresponding to $\psi$). It does not matter whether one
simply identifies $A$ with a stimulus, or interprets $A$ as a kind
of internal representation thereof, while interpreting the stimulus
itself as part of the measurement procedure (together with the instructions
and experimental set-up, that are usually fixed for the entire sequence
of trials). What is important is that the stimulus determines the
observable $A$ uniquely, so that if the same stimulus is presented
in two different trials $t$ and $t'$, one can assume that $A$ is
the same in both of them.%
\footnote{This approach is adopted here to unify it formally with psychological
measurements and quantum measurements in physics (such as spin measurements,
mentioned below). However, one of the authors (JRB) prefers another
approach, adopted in Ref. \cite{Jerome5}, in which stimuli are mapped
into different states and the observable is fixed by the question
about the stimuli. The two approaches coincide with regard to the
measurement issues addressed in this article, and therefore the analysis
is not affected. The relation between the two approaches is outside
the scope of this paper. %
} 

\textcolor{black}{The state $\psi^{\left(t+1\right)}$ determined
by (}\ref{eq: very general rho}\textcolor{black}{) may remain unchanged
between the response $v$ terminating trial $t$ and the presentation
of (the stimulus corresponding to) the new observable that initiates
trial $t+1$. In some applications this interval can indeed be negligibly
small or even zero, but if it is not, one has to allow for the evolution
of $\psi^{\left(t+1\right)}$ within it. In QM, the ``pure'' evolution
of the state (assuming no intervening inter-trial inputs) is described
by some function
\begin{equation}
\psi_{\Delta}^{\left(t+1\right)}=H\left(\psi^{\left(t+1\right)},\Delta\right),\label{eq: very general evolution}
\end{equation}
where $\Delta$ is the time interval between the recording of $v$
in trial $t$ and the observable in trial $t+1$.} This scheme is
somewhat simplistic: one could allow $H$ to depend, in addition to
the time interval $\Delta$, on the observable $A$ and the outcome
$v$ in trial $t$. We do not consider such complex inter-trial dynamics
schemes in this paper.

The reason we single out opinion-polling and compare it to psychophyscis
is that they exemplify two very different types of stimulus-response
relations.

In a typical opinion-polling experiment, a group of participants is
asked one question at a time, e.g., $a=$``Is Bill Clinton honest
and trustworthy?'' and $b=$``Is Al Gore honest and trustworthy?''
\cite{MOO}. The two questions, obviously, differ from each other
in many respects, none of which has anything to do with their content:
the words ``Clinton'' and ``Gore'' sound different, and the participants
know many aspects in which Clinton and Gore differ, besides their
honesty or dishonesty. Therefore, if a question, say, $b$, were presented
to a participant more than once, she would normally recognize that
it had already been asked, which in turn would compel her to repeat
it, unless she wants to contradict herself. One can think of situations
when the respondent can change her opinion, e.g., if another question
posed between two replications of the question provides new information
or reminds something forgotten. Thus, if the answer to the question
$a=$``Do you want to eat this chocolate bar?'' is Yes, and the
second question is $b=$``Do you want to lose weight?,'' the replications
of $a$ may very well elicit response No. It is even conceivable that
if one simply repeats the chocolate question twice, the person will
change her mind, as she may think the replication of the question
is intended to make her ``think again.'' In a wide class of situations,
however, changing one's response would be highly unexpected and even
bizarre (e.g., replace $a$ in the example above with ``Do you like
chocolate?''). We assume that the pairs of questions asked, e.g.,
in Moore's study \cite{MOO} are of this type.

In a typical psychophysical task, the stimuli used are identical in
all respects except for the property that a participant is asked to
judge. Consider a simple detection paradigm in which the observer
is presented one stimulus at a time, the stimulus being either $a$
(containing a signal to be detected) or $b$ (the ``empty'' stimulus,
in which the signal is absent). For instance, $a$ may be a tilted
line segment, and $b$ the same line segment but vertical, the tilt
(which is the signal to be detected) being too small for all answers
to be correct. Clearly, the participant in such an experiment cannot
first decide that the stimulus being presented now has already been
presented before, and that it has to be judged to be $a$ because
so it was before.

With this distinction in mind, however, the formalism (\ref{eq: very general v})-(\ref{eq: very general rho})-(\ref{eq: very general evolution})
can be equally applied to both types of situations. In both cases
$a$ is to be replaced with some observable $A$, and $b$ with some
observable $B$ (after which $a$ and $b$ per se can be forgotten).
The values of $A$ and $B$ are the possible responses one records.
In the psychophysical example, $v\left(A\right)$ and $v\left(B\right)$
each can attain one of two values: $1=$``I think the stimulus was
tilted'' or $0=$``I think the stimulus was vertical''. The psychophysical
analysis consists in identifying the hit-rate and false-alarm-rate
functions (conditioned on the previous stimuli and responses) 
\begin{equation}
\begin{array}{c}
\Pr\left[v\left(A\right)=1\textnormal{ in trial }t\,\vert\,\mbox{measurements in trials }1,\ldots,t-1\right]=F\left(\psi^{\left(t\right)},A,1\right),\\
\Pr\left[v\left(B\right)=1\textnormal{ in trial }t\,\vert\,\mbox{measurements in trials }1,\ldots,t-1\right]=F\left(\psi^{\left(t\right)},B,1\right).
\end{array}\label{eq: binary general}
\end{equation}
The learning (or sequential-effect) aspect of such analysis consists
in identifying the function 
\begin{equation}
\psi^{\left(t+1\right)}=G\left(\psi^{\left(t\right)},S,v\right),\; S\in\left\{ A,B\right\} ,v\in\left\{ 0,1\right\} ,\label{eq: learning binary general}
\end{equation}
combined with the ``pure'' inter-trial dynamics (\ref{eq: very general evolution}).

In the opinion-polling example (say, about Clinton's and Gore's honesty),
there are two hypothetical observables: $A$, corresponding to the
question $a=$``Is Bill Clinton honest?'', and $B$, corresponding
to the question $b=$``Is Al Gore honest?'', each observable having
two possible values, $0=$``Yes'' and $1=$``No''. The analysis,
formally, is precisely the same as above, except that one no longer
uses the terms ``hits'' and ``false alarms'' (because ``honesty''
is not a signal objectively present in one of the two politicians
and absent in another).

It is worth noting that in the opinion polling the observables $A,B$
are defined by the questions $a,b$ alone only because the allowable
responses (Yes or No) and the instructions (``Respond to this question'')
do not vary from one trial to another. If the allowable responses
varied (e.g., if they were Yes and No in some trials, and Yes, No,
and Not Sure in other trials), or if the instruction varied (say,
in some trials ``Respond as quickly as possible'', in other trials
``Think carefully and respond''), they would have also contributed
to the identification of the observables. Analogously, in our psychophysics
example, the observables are defined by stimuli alone because the
instruction to the participants (``Tell us whether the line you see
is tilted or vertical'') and the responses allowed (``Tilted''
and ``Vertical'') remain fixed throughout the successive trials. 

In quantum physics, a classical example falling within the same formal
scheme as the examples above is one involving measuring the spin of
a particle in a given direction. Let the experimenter choose one of
two possible directions, $a$ or $b$ (unit vectors in space along
which the experimenter sets a spin detector). If the particle is a
spin-$\nicefrac{1}{2}$ one, such as an electron, then the spin for
each direction chosen can have one of two possible values, $1=$``up''
or $0=$``down'' (we need not discuss the physical meaning of these
designations). These 1 and 0 are then the possible values of the observables
$A$ and $B$ one associates with the two directions, and the analysis
again consists in identifying the functions $F$, $G$, and $H$.

\section{A brief account of conventional QM\label{sec:A-brief-account}}

In QM, all entities operate in a \emph{Hilbert space}, a vector space
endowed with the operation of scalar product. The components of the
vectors are \emph{complex numbers}. We will assume that the Hilbert
spaces to be considered are $n$-dimensional ($n\geq2$), but the
generalization of all our considerations to infinite-dimensional spaces
is trivial. The scalar product of vectors $\psi,\phi$ is denoted
\[
\left\langle \psi,\phi\right\rangle =\sum_{i=1}^{n}x_{i}y_{i}^{*},
\]
where $x_{i}$ and $y_{i}$ are components of $\psi$ and $\phi$,
respectively, and the star indicates \emph{complex conjugation}: if
$c=a+ib$, then $c^{*}=a-ib$. The length of a vector $\phi$ is defined
as 
\[
\Vert\phi\Vert=\sqrt{\left\langle \phi,\phi\right\rangle }.
\]

Any observable $A$ in this $n$-dimensional version of QM is represented
by an $n\times n$ \emph{Hermitian matrix}.%
\footnote{Each matrix represents an \emph{operator}, in the sense that it transforms
a vector by which it is multiplied into another vector. For this reason,
the terms ``matrix'' and ``operator'' are treated as synonymous
(in a finite-dimensional Hilbert space).%
} This is a matrix with complex entries such that, for any $i,j\in\left\{ 1,\ldots,n\right\} $,
$a_{ij}=a_{ji}^{*}$. In particular, all diagonal entries of $A$
are real numbers. For $n=2$, an observable has the form 
\[
A=\left(\begin{array}{cc}
r & x-iy\\
x+iy & s
\end{array}\right),
\]
where $r,s,x,y$ are real numbers.

It is known from matrix algebra that any Hermitian matrix can be uniquely
decomposed as 
\begin{equation}
A=\sum_{i=1}^{k}v_{i}P_{i},\; k\leq n,\label{eq: spectral}
\end{equation}
where $v_{1},\ldots,v_{k}$ are pairwise distinct \emph{eigenvalues}
of $A$ (all real numbers), and $P_{i}$ are \emph{eigenprojectors}
($n\times n$\emph{ }Hermitian matrices whose eigenvalues are zeros
and ones). For any distinct $i,j\in\left\{ 1,\ldots,k\right\} $,
the eigenprojectors satisfy the conditions 
\begin{equation}
P_{i}^{2}=P_{i}\textnormal{ (idempotency), }P_{i}P_{j}=\mathbf{0}\textnormal{ (orthogonality)}.\label{eq: idempotency, orthogonality}
\end{equation}
Moreover, all eigenprojectors are \emph{positive semidefinite}: for
any nonzero vector $x$, $\langle P_{i}x,x\rangle\geq0$. Finally,
all eigenprojectors sum to the identity matrix, 
\begin{equation}
\sum_{i=1}^{k}P_{i}=I.
\end{equation}

In QM, the distinct eigenvalues $v_{1},\ldots,v_{k}$ are postulated
to form the set of all possible values $v\left(A\right)$. That is,
as a result of measuring $A$ in any given trial one always observes
one of the values $v_{1},\ldots,v_{k}$. For simplicity (and because
all our examples involve binary outcomes), in this paper we will only
deal with the observables $A$ that have two possible values $v\left(A\right)$,
denoted $0$ and $1$. This means that all our observables can be
presented as 
\begin{equation}
A=P_{1},\label{eq: binary observable}
\end{equation}
and 
\begin{equation}
P_{0}^{2}=P_{0},\; P_{1}^{2}=P_{1},\; P_{0}P_{1}=\mathbf{0},\; P_{0}+P_{1}=I.\label{eq: binary properties}
\end{equation}

Each eigenvalue $v$ (0 or 1) has its \emph{multiplicity} $1\leq d<n$.
This is the dimensionality of the \emph{eigenspace} $V$ associated
with $v$, which is the space spanning the $d$ pairwise orthogonal
\emph{eigenvectors} associated with $v$ (i.e., the space of all linear
combinations of these eigenvectors). Multiplication of $P_{v}$ by
any vector $x$ is the orthogonal projection of this vector into $V$.
If $d=1$, the eigenspace $V$ is the ray containing a unique unit-length\emph{
}eigenvector of $A$ corresponding to $v$. The eigenvalue $1-v$
has the multiplicity $n-d$, the dimensionality of the eigenspace
$V^{\bot}$ which is orthogonal to $V$. If both $d=1$ and $n-d=1$
(i.e., $n=2$), then $A$ is said to have a \emph{non-degenerate spectrum}.
In this paper we assume the spectra are generally \emph{degenerate}
($n\geq2$).

The eigenvalues $0,1$ of $A$ in a given trial generally cannot be
predicted, but one can predict the probabilities of their occurrence.
To compute these probabilities, QM uses the notion of a \emph{state}
of the system. In any given trial the state is unique, and it is represented
by a unit length \emph{state vector} $\psi$.%
\footnote{For simplicity, we assume throughout the paper that the system is
always in a \emph{pure state}. A more general \emph{mixed state} is
represented by a \emph{density matrix}, which is essentially a set
of up to $n$ distinct pure states occurring with some probabilities.
The same as with the assumption that $n$ is finite, the restriction
of our analysis to pure states is not critical. %
} If the system is in a state $\psi^{\left(t\right)}$ in trial $t$,
and the measurement is performed on the observable $A$, the probabilities
of the outcomes of this measurement are given by 
\begin{equation}
F\left(\psi^{\left(t\right)},A,v\right)=\Pr\left[v\left(A\right)=v\textnormal{ in trial }t\,\vert\,\mbox{measurements in trials }1,\ldots,t-1\right]=\langle P_{v}\psi^{\left(t\right)},\psi^{\left(t\right)}\rangle=\left\Vert P_{v}\psi^{\left(t\right)}\right\Vert ^{2},\label{eq: probabilities computed}
\end{equation}
where $v=0,1$. Note that these probabilities are conditioned on the
previous observables, in trials $1,\ldots,t-1$, and their observed
values. 

Given that the observed outcome in trial $t$ is $v$, the state $\psi^{\left(t\right)}$
changes into $\psi^{\left(t+1\right)}$ according to 
\begin{equation}
G\left(\psi^{\left(t\right)},A,v\right)=\frac{P_{v}\psi^{\left(t\right)}}{\Vert P_{v}\psi^{\left(t\right)}\Vert}=\psi^{\left(t+1\right)}.\label{eq: projection postulate}
\end{equation}
This equation represents the von Neumann-Lüders \emph{projection postulate}
of QM. The denominator is nonzero because it is the square root of
$\Pr\left[v\left(A\right)=v\textnormal{ in trial }t\right]$, and
(\ref{eq: projection postulate}) is predicated on $v$ having been
observed. The geometric meaning of $G\left(\psi^{\left(t\right)},A,v\right)$
is that $\psi^{\left(t\right)}$ is orthogonally projected by $P_{v}$
into the eigenspace $V$ and then normalized to unit length.

Finally, the inter-trial dynamics of the state vector in QM (between
$v$ and the next observable, separated by interval $\Delta$) is
represented by the \emph{unitary evolution} formula
\begin{equation}
H\left(\psi^{\left(t+1\right)},\Delta\right)=U_{\Delta}\psi^{\left(t+1\right)}=\psi_{\Delta}^{\left(t+1\right)},\label{eq: unitary evolution}
\end{equation}
where $U_{\Delta}$ is a \emph{unitary matrix}, defined by the property
\begin{equation}
U_{\Delta}^{-1}=U_{\Delta}^{\dagger}.
\end{equation}
Here, $U_{\Delta}^{-1}$ is the \emph{matrix inverse} ($U_{\Delta}^{-1}U_{\Delta}=U_{\Delta}U_{\Delta}^{-1}=I$),
and \textcolor{black}{$U_{\Delta}^{\dagger}$ is the }\textcolor{black}{\emph{conjugate
transpose}}\textcolor{black}{{} of $U_{\Delta}$, obtained by transposing
$U_{\Delta}$ and replacing each entry $x+iy$ in it with its complex
conjugate $x-iy$. }The unitary matrix $U_{\Delta}$ should also be
made a function of inter-trial variations in the environment (such
as variations in overall noise level, or other participants' responses)
if they are non-negligible. The identity matrix $I$ is a unitary
matrix: if $U_{\Delta}=I$, (\ref{eq: unitary evolution}) describes
\emph{no inter-trial dynamics}, with the state remaining the same
through the interval $\Delta$. 

Note that the eigenvalue $v$ itself does not enter the computations.
This justifies treating it as merely a label for the eigenprojectors
and eigenspaces (so instead of $0,1$ we could use any other labels).

\section{Measurement sequences, evolution (in)effectiveness, and stability}

In this section we introduce terminology and preliminary considerations
needed in the subsequent analysis. Throughout the paper we will make
use of the following way of describing measurements performed in successive
trials: 
\begin{equation}
\left(A_{1},v_{1},p_{1}\right)\rightarrow\ldots\rightarrow\left(A_{r},v_{r},p_{r}\right).
\end{equation}
We call this a \emph{measurement sequence}. Each triple in the sequence
consists of an observable $A$ being measured, an outcome $v$ recorded
(0 or 1), and its conditional probability $p$. The probability is
conditioned on the observables measured and the outcomes recorded
in the previous trials of the same measurement sequence. Thus, 
\begin{equation}
\begin{array}{c}
p_{1}=\Pr\left[v\left(A_{1}\right)=v_{1}\textnormal{ in trial }1\right],\\
p_{2}=\Pr\left[v\left(A_{2}\right)=v_{2}\textnormal{ in trial }2\:\vert\: v\left(A_{1}\right)=v_{1}\textnormal{ in trial }1\right],\\
p_{2}=\Pr\left[v\left(A_{3}\right)=v_{3}\textnormal{ in trial }3\:\vert\: v\left(A_{1}\right)=v_{1}\textnormal{ in trial }1,\textnormal{ and }v\left(A_{2}\right)=v_{2}\textnormal{ in trial }2\right],\\
\ldots
\end{array}
\end{equation}
As we assume that the outcomes $v_{1},v_{2},\ldots$ in a measurement
sequence have been recorded, all probabilities $p_{1},p_{2},\ldots$
are positive if the measurement sequence exists. Recall that the observables
$A_{1},A_{2},\ldots$ in a sequence are uniquely determined by the
measurement procedures applied, $a_{1},a_{2},\ldots$, and that the
outcomes (0 or 1) are eigenvalues of these observables.

Consider now the two-trial measurement sequence
\begin{equation}
\left(A,v,p\right)\rightarrow\left(B,w,q\right),\label{eq: two trials}
\end{equation}
where $v,w\in\left\{ 0,1\right\} $. Let $A$ have the eigenprojector
matrices $P_{0},P_{1}$, and $B$ have the eigenprojector matrices
$Q_{0},Q_{1}$. If the initial state of the system is $\psi=\psi^{\left(1\right)}$,
we have 
\begin{equation}
p=\Vert P_{v}\psi\Vert^{2},
\end{equation}
and $\psi^{\left(1\right)}$ transforms into
\begin{equation}
\frac{P_{v}\psi}{\Vert P_{v}\psi\Vert}=\psi^{\left(2\right)}.
\end{equation}
Assuming an interval $\Delta$ between the two trials, $\psi^{\left(2\right)}$
evolves into 
\begin{equation}
\psi_{\Delta}^{\left(2\right)}=U_{\Delta}\psi^{\left(2\right)}=\frac{U_{\Delta}P_{v}\psi}{\Vert P_{v}\psi\Vert}.
\end{equation}
This is the state vector paired with $B$ in the next measurement,
yielding, with the help of some algebra,
\begin{equation}
q=\Vert Q_{w}\psi_{\Delta}^{\left(2\right)}\Vert^{2}=\frac{\Vert Q_{w}U_{\Delta}P_{v}\psi\Vert^{2}}{\Vert P_{v}\psi\Vert^{2}}=\frac{\Vert\left(U_{\Delta}^{\dagger}Q_{w}U_{\Delta}\right)P_{v}\psi\Vert^{2}}{\Vert P_{v}\psi\Vert^{2}}.\label{eq: q with dynamics}
\end{equation}
As a special case $U_{\Delta}$ can be the identity matrix (no inter-trial
changes in the state vector), and then we have
\begin{equation}
q=\frac{\Vert Q_{w}P_{v}\psi\Vert^{2}}{\Vert P_{v}\psi\Vert^{2}},\label{eq: q simple}
\end{equation}
because in this case
\begin{equation}
\left(U_{\Delta}^{\dagger}Q_{w}U_{\Delta}\right)=Q_{w}.\label{eq: inefficiency}
\end{equation}

It is possible, however, that the latter equality holds even if $U_{\Delta}^{\dagger}$
is not the identity matrix. In fact it is easy to see that this happens
if and only if $U_{\Delta}$ and $B$ commute, i.e.,
\begin{equation}
U_{\Delta}B=BU_{\Delta}.\label{eq: commutation with evolution}
\end{equation}
For the proof of this, see Lemma \ref{lem: inefficient} in Appendix. 

We will say that 
\begin{defn}
A unitary operator $U_{\Delta}$ is \emph{ineffective for an observable}
$B$ if the two operators commute. 
\end{defn}
The justification for this terminology should be transparent: due
to (\ref{eq: inefficiency}), in the computation (\ref{eq: q with dynamics})
of the probability $q$ the evolution operator can be ignored, yielding
(\ref{eq: q simple}). The notion of inefficiency of the evolution
operator will play an important role in the analysis of repeated measurements
below.

Our next consideration regards the set of all possible values of the
initial state vector $\psi$ for a given measurement sequence. In
the applications of QM in physics, this set is assumed to cover the
entire Hilbert space in which they are defined. We are not justified
to adopt this assumption in psychology, it would be too strong: one
could argue that the initial states in a given experiment may be forbidden
to attain values within certain areas of the Hilbert space. At the
same time, it seems even less reasonable to allow for the possibility
that the initial state for a given measurement sequence is always
fixed at one particular value. The initial state vectors, as follows
from both the QM principles and common sense, should depend on the
system's history prior to the given experiment, and this should create
some variability from one replication of this experiment to another.
This is important, because, given a set of observables, specially
chosen initial state vectors may exhibit ``atypical'' behaviors,
those that would disappear if the state vector were modified even
slightly.

This leads us to adopting the following\medskip{}

\textbf{Stability Principle:} \emph{If $\psi$ is a possible initial
state vector for a given measurement sequence in an $n$-dimensional
Hilbert space, then there is an open ball $B_{r}\left(\psi\right)$
centered at $\psi$ with a sufficiently small radius $r$, such that
any vector $\psi+\delta$ in this ball,}\textcolor{black}{{} }\textcolor{black}{\emph{normalized
by its length $\left\Vert \psi+\delta\right\Vert $,}}\emph{ is also
a possible initial state vector for this measurement sequence.}\medskip{}

\noindent We will say that 
\begin{defn}
A property of a measurement sequence is (or holds)\emph{ stable} for
an\emph{ }initial vector $\psi$, if it holds for all state vectors
within a sufficiently small $B_{r}\left(\psi\right)$. 
\end{defn}
Almost all our propositions below are proved under this stability
clause, specifically by using the reasoning presented in Lemma \ref{lem: stability}
in Appendix.

\section{Consequences for ``a\textmd{\normalsize{$\rightarrow$}}a{\normalsize{''-type
measurement sequences}}}

Using the definitions and the language just introduced, we will now
focus on the consequences of (\ref{eq: probabilities computed})-(\ref{eq: projection postulate})-(\ref{eq: unitary evolution})
for repeated measurements with repeated responses,
\begin{equation}
\left(A,v,p\right)\rightarrow\left(A,v,p'\right).\label{eq: repeated}
\end{equation}
Consider an opinion-polling experiment, with questions like $a=$``Is
Bill Clinton trustworthy?'' \cite{MOO}. As argued for in Introduction,
if the same question is posed twice, $a\rightarrow a$, a typical
respondent, who perhaps hesitated when choosing the response the first
time she was asked $a$, would now certainly be expected to repeat
it, perhaps with some display of surprise at being asked the question
she has just answered. This may not be true for all possible questions,
but it is certainly true for a vast class thereof. Let us formulate
this as

\paragraph{Property 1\label{Fact zero}}

\noindent \emph{For some nonempty class of questions, if a question
is repeated twice in successive trials (separated by one of a broad
range of inter-trial intervals), the response to it will also be repeated}.\medskip{}

If a question $a$ within the scope of this property is represented
by an observable $A$, we are dealing with the measurement sequence
(\ref{eq: repeated}) in which $p'=1$. Such a measurement sequence
does not disagree with the formulas (\ref{eq: probabilities computed})-(\ref{eq: projection postulate})-(\ref{eq: unitary evolution}),
and in fact is even predicted by them if the intervening inter-trial
evolution of the state vector is assumed to be ineffective. Indeed,
(\ref{eq: q with dynamics}) for the measurement sequence (\ref{eq: repeated})
acquires the form 
\begin{equation}
p'=\frac{\Vert\left(U_{\Delta}^{\dagger}P_{v}U_{\Delta}\right)P_{v}\psi\Vert^{2}}{\Vert P_{v}\psi\Vert^{2}},\label{eq: p' ineffiecient evolution}
\end{equation}
and the inefficiency of $U_{\Delta}$ for $A$ implies
\begin{equation}
p'=\frac{\Vert P_{v}^{2}\psi\Vert^{2}}{\Vert P_{v}\psi\Vert^{2}}=1,\label{eq: p' again}
\end{equation}
because $P_{v}^{2}=P_{v}$ holds for all projection operators.

This is easy to understand informally. An outcome $v$ observed in
the first measurement, $\left(A,v,p\right)$, is associated with an
eigenspace $V$. The measurement orthogonally projects the state vector
$\psi=\psi^{\left(1\right)}$ into this eigenspace, and this projection
is normalized to become the new state $\psi^{\left(2\right)}$. The
application of the same measurement to $\psi^{\left(2\right)}$ orthogonally
projects it into $V$ again, but since $\psi^{\left(2\right)}$ is
already in $V$, it does not change. The squared length of the projection
therefore is $1$, and this is what the probability $p'$ is.

As it turns out, under the stability principle, effective inter-trial
evolution is in fact excluded for the observables representing the
questions falling within the scope of Property 1. In other words,
for all such questions, the unitary operators $U_{\Delta}$ can be
ignored in all probability computations. 

Let us say that 
\begin{defn}
An observable $A$ has the \emph{Lüders property} with respect to
a state vector $\psi$ if the existence of the measurement $\left(A,v,p\right)$
for this $\psi$ and an outcome $v\in\left\{ 0,1\right\} $ implies
that the property $p'=1$ holds stable for this $\psi$ in the measurement
sequence $\left(A,v,p\right)\rightarrow\left(A,v,p'\right)$. 
\end{defn}
In other words, the Lüders property means that an answer to a question
(represented by $A$) is repeated if the question is repeated, and
that this is true not just for one initial state vector $\psi$, but
for all state vectors sufficiently close to it.

We now can formulate our first proposition.
\begin{prop}[repeated measurements]
\emph{}\label{prop: aa dynamics} An observable $A$ has the Lüders
property if and only if $U_{\Delta}$ in (\ref{eq: unitary evolution})
is ineffective for $A$.
\end{prop}
See Appendix for a formal proof. In the formulation of Property 1,
the interval $\Delta$ and the question represented by $A$ can vary
within some broad limits, whence the inefficiency of $U_{\Delta}$
for $A$ should also hold for each of these intervals combined with
each of these questions. 

We have to be careful not to overgeneralize the Lüders property and
the ensuing inefficiency property. As we discussed in Introduction,
one can think of situations where replications of a question may lead
the respondent to ``change her mind.'' The most striking contrast,
however, is provided by psychophysical applications of QM. Here, the
inter-trial dynamics not only cannot be ignored, it must play a central
role. 

Let us illustrate this on an old but very thorough study by Atkinson,
Carterette, and Kinchla \cite{ACK1}. In the experiments they report,
each stimulus consisted of two side-by-side identical fields of luminance
$L$, to one of which a small luminance increment $\Delta L$ could
be added, serving as the signal to be detected. There were three stimuli:
\begin{equation}
a=\left(L+\Delta L,L\right),\: b=\left(L,L+\Delta L\right),\: c=\left(L,L\right).
\end{equation}
In each trial the observer indicated which of the two fields, right
one or left one, contained the signal. There were thus two possible
responses: Left and Right. An application of QM analysis to these
experiments requires $a,b,c$ to be translated into observable $A,B,C$,
each with two eigenvalues, say, $0=\textnormal{Left}$ and $1=\textnormal{Right}$.
In the experiments we consider no feedback was given to the observers
following a response. This is a desirable feature. It makes the sequence
of trials we consider formally comparable to successive measurements
of spins in quantum physics: measurements simply follow each other,
with no interventions in between.%
\footnote{However, this precaution seems unnecessary, as the results of the
experiments with feedback in Ref. \cite{ACK1} do not qualitatively
differ from the ones we discuss here. %
}

We are interested in measurement sequences 
\begin{equation}
\begin{array}{cc}
\left(A,0,p_{1}\right)\rightarrow\left(A,0,p_{1}'\right), & \left(A,1,p_{2}\right)\rightarrow\left(A,1,p_{2}'\right),\\
\left(B,0,p_{3}\right)\rightarrow\left(B,0,p_{3}'\right), & \left(B,1,p_{4}\right)\rightarrow\left(B,1,p_{4}'\right),\\
\left(C,0,p_{5}\right)\rightarrow\left(C,0,p_{5}'\right), & \left(C,1,p_{6}\right)\rightarrow\left(C,1,p_{6}'\right).
\end{array}\label{eq: experiment ACK}
\end{equation}
Recall that the probabilities $p'_{i}$ ($i=1,\ldots,6$) are conditioned
on previous measurements, so that, e.g., $p'_{1}+p'_{2}\not=1$ while
$p_{1}+p_{2}=1$. 

For each observer, the probabilities were estimated from the last
400 trials out of 800 (to ensure an ``asymptotic'' level of performance).
The results, averaged over 24 observers, were as follows:

\[
\boxed{\begin{array}{c}
\textnormal{Experiment 1}\\
\begin{array}{ccc}
\textnormal{index } & p & p'\\
1 & .65 & .73\\
2 & .35 & .38\\
3 & .36 & .39\\
4 & .64 & .71\\
5 & .50 & .53\\
6 & .50 & .60
\end{array}
\end{array}}\quad\boxed{\begin{array}{c}
\textnormal{Experiment 2}\\
\begin{array}{ccc}
\textnormal{index } & p & p'\\
1 & .56 & .70\\
2 & .44 & .41\\
3 & .27 & .31\\
4 & .73 & .79\\
5 & .39 & .50\\
6 & .61 & .65
\end{array}
\end{array}}
\]
The leftmost column in each table corresponds to the index associated
with $p$ and $p'$ in (\ref{eq: experiment ACK}). Thus, the first
row shows $p_{1}$ and $p'_{1}$, the last one shows $p_{6}$ and
$p'_{6}$. The two experiments differed in one respect only: in Experiment
1 the stimuli $a$ and $b$ were presented with equal probabilities,
while in Experiment 2 the stimulus $b$ was three times more probable
that $a$ (the probability of $c$ was 0.2 in both experiments). The
results show, in accordance with conventional detection models, that
this manipulation makes responses in Experiment 2 biased towards the
correct response to $b$. This aspect of the data, however, is not
of any significance for us. What is significant, is that, in accordance
with Proposition \ref{prop: aa dynamics}, we should conclude that
the inter-trial evolution (\ref{eq: unitary evolution}) here intervenes
always and significantly.

\section{A consequence for ``a$\rightarrow$b$\rightarrow$a''-type measurement
sequences }

Returning to the opinion polling experiments, consider the situation
involving two questions, such as $a=$``Is Bill Clinton honest?''
and $b=$``Is Al Gore honest?'' The two questions are posed in one
of the two orders, $a\rightarrow b$ or $b\rightarrow a$, to a large
group of people. The same as with asking the same question twice in
a row, one would normally consider it unnecessary to extend these
sequences by asking one of the two questions again, by repeating $b$
or $a$ after having asked $a$ and $b$. A typical respondent, again,
will be expected to repeat her first response. We find it ``almost
certain'' (the ``almost'' being inserted here because we cannot
refer to any systematic experimental study of this obvious expectation)
that from the nonempty (in reality, vast) class of questions falling
within the scope of Property 1 one can always choose pairs of questions
falling within the scope of the following extension of this property.

\paragraph{Property 2.\label{Fact 2}}

\emph{Within a nonempty subclass of questions (and for the same set
of inter-trial intervals) for which Property 1 holds, if a question
$a$ is asked following questions $a$ and $b$ (in either order),
the response to it will necessarily be the same as that given to the
question $a$ the first time it was asked. \medskip{}
}

As always, we replace $a,b$ with observables $A,B$, and use the
following notation: the probability of obtaining a value $v$ when
measuring the observable $A$ is denoted $p_{vA}$, $q_{vA}$, etc.
(the letters $p,q,$ etc. distinguishing different measurements);
we use analogous notation for the probability of obtaining a value
$w$ when measuring the observable $B$. 

Consider the measurement sequence
\begin{equation}
\left(A,v,p_{vA}\right)\rightarrow\left(B,w,p_{wB}\right)\rightarrow\left(A,v,p'_{vA}\right)\label{eq: ABA}
\end{equation}
Property 2 implies that in these sequences $p'_{vA}=1$ and $q'_{wA}=1$.
As it turns out, this property has an important consequence (assuming
the two inter-trial intervals in the measurement sequences belong
to the same class as $\Delta$ in Proposition \ref{prop: aa dynamics}).
\begin{prop}[alternating measurements]
\emph{}\label{prop:abab} Let $A$ and $B$ possess the Lüders property,
and let the measurement sequences \textup{ 
\begin{equation}
\left(A,v,p_{vA}\right)\rightarrow\left(B,w,p_{wB}\right)\label{eq: AB}
\end{equation}
}exist for all $v,w\in\left\{ 0,1\right\} $, and some initial state
vector $\psi$. Then, in the measurement sequences (\ref{eq: ABA}),\textup{
}the property $p'_{vA}=1$ holds stable for this $\psi$ if and only
if $A$ and $B$ commute, $AB=BA.$ 
\end{prop}
In other words, if the probabilities $p_{vA},p_{wB},q_{wB},q_{vA}$
are nonzero in (\ref{eq: AB}) for some $\psi$, the sequences (\ref{eq: ABA})
exist with $p'_{vA}=1$ and $q'_{wA}=1$ for all state vectors in
a small neighborhood of $\psi$ if and only if $AB=BA$. See Appendix
for a formal proof.

To illustrate how this works, recall that $A$ and $B$ commute if
and only if they have one and the same set of orthonormal eigenvectors
$e_{1},\ldots,e_{n}$ (generally, not unique). Since $A$ and $B$
have two eigenvalues each, the difference between the two observables
is in how these eigenvectors are grouped into two eigenspaces. Take
one of the measurement sequences of the (\ref{eq: ABA})-type, say,
\begin{equation}
\left(A,1,p_{1a}\right)\rightarrow\left(B,0,p_{0B}\right)\rightarrow\left(A,1,p'_{1a}\right).
\end{equation}
Since $A$ and $B$ have the Lüders property, all the probabilities
are the same as if there was no inter-trial dynamics involved. Proceeding
under this assumption, the first measurement projects the initial
vector $\psi=\psi^{\left(1\right)}$ into $V_{1}$ that spans some
of the vectors $e_{1},\ldots,e_{n}$. Let this projection (after its
length was normalized to 1) be $\psi^{\left(2\right)}$. The second
measurement projects $\psi^{\left(2\right)}$ into the intersection
$V_{1}\cap W_{0}$ that spans a smaller subset of these vectors. The
third measurement then, since the second normalized projection $\psi^{\left(3\right)}$
is already in $V_{1}$, does not change it, $\psi^{\left(4\right)}=\psi^{\left(3\right)}$.
This means that the third probability, $p'_{1a}$, being the scalar
product of $\psi^{\left(3\right)}$ and $\psi^{\left(4\right)}$,
must be unity.

The commutativity of $A$ and $B$ is important because it has an
experimentally testable consequence. 
\begin{prop}[no order effect]
\label{prop: commutativity} If $A$ and $B$ possessing the L\"uders
property commute, then in the measurement sequences 
\[
\left(A,v,p_{vA}\right)\rightarrow\left(B,w,p_{wB}\right),
\]
\[
\left(B,w,q_{wB}\right)\rightarrow\left(A,v,q_{vA}\right)
\]
the joint probabilities of the two outcomes are the same, 
\begin{equation}
p_{vA}p_{wB}=q_{wB}q_{vA}.\label{eq: no order effect}
\end{equation}
Consequently, 
\begin{equation}
\Pr\left[v\left(A\right)=v\textnormal{ in trial 1}\right]=\Pr\left[v\left(A\right)=v\textnormal{ in trial 2}\right]\label{eq: no order 1}
\end{equation}
and 
\begin{equation}
\Pr\left[w\left(B\right)=w\textnormal{ in trial 1}\right]=\Pr\left[w\left(B\right)=w\textnormal{ in trial 2}\right].\label{eq: no order 2}
\end{equation}

\end{prop}
To clearly understand what is being stated, recall that $p_{wB}$
is the conditional probability of observing the value $w$ of $B$
given that before this the outcome was the value $v$ of $A$. So,
the product of $p_{vA}p_{wB}$ is the overall probability of the first
of the two sequences in the proposition. The value of $q_{wB}q_{vA}$
is understood analogously. Equation (\ref{eq: no order effect}) therefore
states that 
\[
\Pr\left[v\left(A\right)=v\textnormal{ in trial 1 and }w\left(B\right)=w\textnormal{ in trial 2}\right]=\Pr\left[w\left(B\right)=w\textnormal{ in trial 1 and }v\left(A\right)=v\textnormal{ in trial 2}\right].
\]
The proof of the proposition is given in Appendix.

Equations (\ref{eq: no order effect})-(\ref{eq: no order 1})-(\ref{eq: no order 2})
are empirically testable predictions. Moreover, if we assume that
the questions like ``Is Clinton honest'' and ``Is Gore honest''
fall within the scope of Property 2 (and it would be amazing if they
did not), these predictions are known to be de facto falsified.

\paragraph{Property 3.\label{Fact 3}}

\emph{Within a nonempty subclass of questions for which Property 2
holds (and for the same set of inter-trial intervals), the joint probability
of two successive responses depends on the order in which the questions
were posed.\medskip{}
} 

This ``\emph{question order effect}'' has in fact been presented
as one for whose understanding QM is especially useful: the empirical
finding that $p_{vA}p_{wB}\not=q_{wB}q_{vA}$ is explained in Ref.
\cite{WANG} by assuming that $A$ and $B$ do not commute. In the
survey reported by Moore \cite{MOO}, about 1,000 people were asked
two questions, one half of them in one order, the other half in another.
The probability estimates are presented for four pairs of questions:
the first pair was about the honesty of Bill Clinton $(a)$ and Al
Gore (\emph{b}), the second about the honesty of Newt Gingrich ($a$)
and Bob Dole ($b$), etc. 
\[
\boxed{\begin{array}{cc}
 & \textnormal{probability of Yes to }a\\
\begin{array}{c}
\textnormal{question pair}\\
1\\
2\\
3\\
4
\end{array} & \begin{array}{cc}
\textnormal{in }a\rightarrow b & \textnormal{ in }b\rightarrow a\\
.50 & .57\\
.41 & .33\\
.41 & .53\\
.64 & .52
\end{array}
\end{array}}\quad\boxed{\begin{array}{cc}
 & \textnormal{probability of Yes to }b\\
\begin{array}{c}
\textnormal{question pair}\\
1\\
2\\
3\\
4
\end{array} & \begin{array}{cc}
\textnormal{in }a\rightarrow b & \textnormal{ in }b\rightarrow a\\
.60 & .68\\
.64 & .60\\
.56 & .46\\
.33 & .45
\end{array}
\end{array}}
\]
The results are presented in the form of $\Pr\left[v\left(A\right)=1\textnormal{ in trial }i\right]$
and $\Pr\left[w\left(B\right)=1\textnormal{ in trial }i\right]$,
$i=1,2$, so the tested predictions are (\ref{eq: no order 1}) and
(\ref{eq: no order 2}). As we can see, for all question pairs, the
probability estimates of Yes to the same question differ depending
on whether the question was asked first or second. Given the sample
size (about 500 respondents per question pair in a given order) the
differences are not attributable to chance variation.

Properties 1, 2, and 3 turn
out to be incompatible within the framework of QM. We should conclude
therefore that QM cannot be applied to the questions that have these
properties without modifications.

\section{Would POVMs work?}

Are there more flexible versions (generalizations) of QM that could
be used instead?

One widely used generalization of the conventional QM involves replacing
the projection operators with \emph{positive-operator-valued measures}
(POVMs), see, e.g., Refs. \cite{BUS} and \cite{DMU}. The conceptual
set-up here is as follows. We continue to deal with an $n$-dimensional
Hilbert space ($n\geq2$). The notion of a state represented by a
unit vector $\psi$ in this space remains unchanged. The generalization
occurs in the notion of an observable. For experiments with binary
outcomes, an observable $A$ of the conventional QM is defined by
(\ref{eq: binary observable}), with eigenvalues $\left(0,1\right)$
and eigenprojectors $\left(P_{0},P_{1}\right)$. The eigenvalues themselves
are not relevant insofar as they are distinct: replacing $0,1$ with
another pair of distinct values amounts to trivial relabeling of the
measurement outcomes. The information about the observable $A$ therefore
is contained in the eigenprojectors $P_{0},P_{1}$. They are Hermitian
positive semidefinite operators subject to the restrictions (\ref{eq: binary properties}).

A generalized observable, or POVM, $A$ (continuing to consider only
binary outcomes) is defined as a pair $\left(E_{0},E_{1}\right)$
of Hermitian positive semidefinite operators in the $n$-dimensional
Hilbert space, summing to the identity matrix $I$. In other words,
the generalization from eigenprojectors $P_{v}$ to POVM components
$E_{v}$ amounts to dropping the idempotency and orthogonality constraints,
defined in (\ref{eq: idempotency, orthogonality}).\textcolor{black}{}%
\footnote{Without getting into details, the theory of POVMs is sometimes referred
to as the \emph{open-system} QM because of \emph{Naimark's dilation
theorem} \cite{NEU}. It says that any POVM $A$ in a Hilbert space
$H$ can be represented by a conventional quantum observable $\tilde{A}$
in a Hilbert space of higher dimensionality, the tensor product $H\otimes K$
of the original Hilbert space $H$ and another Hilbert space $K$.
The latter is interpreted as an \emph{environment} for $H$. The measurements
of $\tilde{A}$ are performed in a conventional way in $H\otimes K$,
and the resulting state vectors are projected back into $H$.%
} 

\textcolor{black}{Any component $E_{v}$ ($v=0,1$) can be presented
as $M_{v}^{\dagger}M_{v}$, where $M_{v}$ is some matrix and $M_{v}^{\dagger}$
is its conjugate transpose. The representation $E=M_{v}^{\dagger}M_{v}$
for a given $E_{v}$ is not unique, but it is supposed to be fixed
within a given experiment (i.e., for a given measurement procedure).}

The measurement formulas specifying $F$ and $G$ in (\ref{eq: very general v})-(\ref{eq: very general rho})
can now be formulated to resemble (\ref{eq: probabilities computed})-(\ref{eq: projection postulate}).
The conditional probability of an outcome $v=0,1$ of the measurement
of $A=\left(E_{0},E_{1}\right)$ in state $\psi^{\left(t\right)}$
is 
\begin{equation}
{\color{black}F\left(\psi^{\left(t\right)},A,v\right)=\Pr\left[v\left(A\right)=v\textnormal{ in trial }t\,\vert\,\mbox{measurements in trials }1,\ldots,t-1\right]={\color{black}\left\langle E_{v}\psi^{\left(t\right)},\psi^{\left(t\right)}\right\rangle =\Vert M_{v}\psi^{\left(t\right)}\Vert^{2}}}\label{eq: probabilities computed-POVM}
\end{equation}
This measurement transforms $\psi^{\left(t\right)}$ into \textcolor{black}{
\begin{equation}
{\color{black}{\color{black}{\color{red}{\color{black}G\left(\psi^{\left(t\right)},A,v\right)=\frac{M_{v}\psi^{\left(t\right)}}{\Vert M_{v}\psi^{\left(t\right)}\Vert}=\psi^{\left(t+1\right)}.}}}}\label{eq: projection postulate-POVM}
\end{equation}
The formula for the evolution of the state vector between trials remains
the same as for the conventional observables, (\ref{eq: unitary evolution}).}

It is easy to see that we no longer need to involve inter-trial changes
in the state vector to explain the fact that, in psychophysics, a
replication of stimulus does not lead to the replication of response.
In a measurement sequence 
\[
\left(A,v,p\right)\rightarrow\left(A,v,p'\right),
\]
if $U_{\Delta}$ is the identity matrix, then $p'$ is given by\textcolor{black}{{}
\[
p'=\left\langle E_{v}\frac{M_{v}\psi}{\Vert M_{v}\psi\Vert},\frac{M_{v}\psi}{\Vert M_{v}\psi\Vert}\right\rangle =\frac{\left\langle (M_{v}^{\dagger})^{2}M_{v}^{2}\psi,\psi\right\rangle }{\langle M_{v}^{\dagger}M_{v}\psi,\psi\rangle}.
\]
This value is generally different from 1: since $(M_{v}^{\dagger})^{2}M_{v}^{2}$,
where $M_{v}$ is not necessarily an orthogonal projector, is generally
different from $M_{v}^{\dagger}M_{v}$, $\langle(M_{v}^{\dagger})^{2}M_{v}^{2}\psi,\psi\rangle$
is generally different from $\langle M_{v}^{\dagger}M_{v}\psi,\psi\rangle$. }

This is interesting, as it suggests the possibility of treating psychophysical
judgments and opinion polling within the same (evolution-free) framework.
This encouraging possibility, however, cannot be realized: the theory
of POVMs cannot help us in reconciling Properties 2 and
3 in opinion-polling, because POVMs with Lüders property
cannot be anything but conventional observables. This is shown in
the following
\begin{prop}[no generalization]
\label{propo: POVM} A POVM $A=\left(E_{0},E_{1}\right)$ has the
Lüders property with respect to a state $\psi$ if and only if $A$
is a conventional observable (i.e., it is a Hermitian operator, and
its components $E_{0},E_{1}$ are its eigenprojectors). 
\end{prop}
See Appendix for a proof.

Proposition \ref{propo: POVM} says that POVMs to be used to model
opinion polling should be conventional observables, otherwise Property
1 will be necessarily contradicted. But then Propositions \ref{prop: aa dynamics}
and \ref{prop:abab} are applicable, and they say that the inter-trial
dynamics is ineffective, and that all the observables representing
different questions within the scope of Property 2 pairwise commute.
This, in turn, allows us to invoke Proposition \ref{prop: commutativity},
with the result that, contrary to Property 3, the order of the questions
should have no effect on the response probabilities.

\section{Conclusion}

Let us summarize. Both cognitive science and quantum physics deal
with fundamentally probabilistic input-output relations, exhibiting
a variety of sequential effects. Both deal with these relations and
effects by using, in some form or another, the notion of an ``internal
state'' of a system. In psychology, the maximally general version
is provided by the probabilistic generalization of the old behaviorist
S-O-R scheme: the probability of an output is a function of the input
and the system's current state (function $F$ in (\ref{eq: very general v})),
and both the input and the output change the current state into a
new state (function $G$ in (\ref{eq: very general rho})). If we
discretize behavior into subsequent trials, then we need also a function
describing how the state of the system changes between the trials
(function $H$ in (\ref{eq: very general evolution})).

Quantum physics uses a special form of the functions $F$, $G$, and
$H$, the ones derived from (or constituting, depending on the axiomatization)
the principles of QM. Functions $F$ and $G$ are given by (\ref{eq: probabilities computed})-(\ref{eq: projection postulate})
in the conventional QM, and by (\ref{eq: probabilities computed-POVM})-(\ref{eq: projection postulate-POVM})
in the QM with POVMs, with the inter-trial evolution in both cases
described by (\ref{eq: unitary evolution}). Nothing a priori precludes
these special forms of $F,G,H$ from being applicable in cognitive
science, and such applications were successfully tried: by appropriately
choosing observables and states, certain experimental data in human
decision making were found to conform with QM predictions \cite{POTZ}.

As this paper shows, however, QM encounters difficulties in accounting
for some very basic empirical properties. In opinion polling (more
generally, in all psychological tasks where stimuli/questions can
be confidently identified by features other than those being judged),
there is a class of questions such that a repeated question is answered
in the same way as the first time it was asked. This agrees with the
Lüders projection postulate, and renders the use of both the inter-trial
dynamics of the state vector and the POVM theory unnecessary: to have
this property the questions asked have to be represented by conventional
observables with ineffective inter-trial dynamics. In many situations,
we also expect that for a certain class of questions the response
to two replications of a given question remains the same even if we
insert another question in between and have it answered. This property
can only be handled by QM if the conventional observables representing
different questions all pairwise commute, i.e., can be assigned the
same set of eigenvectors. This, in turn, leads to a strong prediction:
the joint probability of two responses to two successive questions
does not depend on their order. This prediction is known to be violated
for some pairs of questions. The explanation of the ``question order
effect'' is in fact one of the most successful applications of QM
in psychology \cite{WANG}, but it requires noncommuting observables,
and these, as we have seen, cannot account for the repeated answers
to repeated questions.%
\footnote{Some of these issues were previously raised in a more informal manner
by Geoff Iverson (personal communication, June 2010).%
}

Our paper in no way dismisses the applications of QM in cognitive
psychology, or diminishes their modeling value. It merely sounds a
cautionary note: it seems that we lack a deeper theoretical foundation,
a set of well-justified principles that would determine where QM can
and where it must not be used. We should also point out that the problems
identified in this paper are not unique to QM. For example, random
utility theories also have difficulty explaining the trial to trial
dependencies in answers to questions. If we assume, as done in traditional
random utility theories (see, e.g., Ref. \cite{REG}), that a response
is based on a randomly sampled utility in each trial, then repeating
the response will produce different random samples in each trial.
That is why in the experiments designed to test random utility models
questions never repeated back to back, and instead ``filler trials''
are inserted to make participants forget their earlier choice.

Clearly, the basic properties that we have shown to contravene QM
can be ``explained away'' by invoking considerations formulated
in traditional psychological terms. One can, e.g., dismiss the problem
with repeated questions in opinion polling by pointing out that the
respondents ``merely'' remember their previous answers and ``simply''
do not want to contradict themselves. One can similarly dismiss the
question order effect by pointing out that the first question ``simply''
changes the context for pondering the second question, e.g., reminds
something the respondent would not have thought of had the second
question been asked first. These may very well be valid considerations.
But if one allows for such extraneous to QM explanations, one needs
to understand (A) why the same extraneous considerations do not intervene
in situations where QM is successfully applicable, and (B) why one
cannot stick to considerations of this kind and dispense with QM altogether. 

A reasonable answer is that the value of QM applications is precisely
in that it replaces the disparate conventional psychological notions
with unified and mathematically rigorous ones. But then in those situations
where we find QM not applicable one needs more than invoking these
conventional psychological notions. One needs principles. Both in
a psychophysical detection experiment and in opinion polling, participants
may think of various things between trials, and previously presented
stimuli/questions as well as previously given responses definitely
change something in their mind, affecting their responses to subsequent
stimuli/questions. Why then the applicability of QM is not the same
in these two cases? Why, e.g., should the inter-trials dynamics of
the state vector (or the use of POVMs in place of conventional observables)
be critical in one case and ineffective (or unnecessary) in another? 

One should also consider the possibility that rather than acting as
switches distinguishing the situations in which (\ref{eq: probabilities computed})-(\ref{eq: projection postulate})
or (\ref{eq: probabilities computed-POVM})-(\ref{eq: projection postulate-POVM})
are and are not applicable, the set of the hypothetical principles
in question may require a higher level of generality for the functions
$F,G,H$. A serious and meticulous work is needed therefore to determine
precisely what features of QM are critical for this or that (un)successful
explanation. As an example, virtually any functions $F,G,H$ in the
general formulas (\ref{eq: very general v})-(\ref{eq: very general rho})-(\ref{eq: very general evolution})
predict the existence of the question order effect, and the functions
can always be adjusted to account for any specific effect. The QQ
constraint for the question order effect discovered by Wang and Busemeyer
\cite{WANG} means that, for any two questions $a,b$ and any respective
responses $v,w\in\left\{ 0,1\right\} $, 
\[
h_{ab}\left(v,w\right)=\Pr\left[v\textnormal{ in response to }a\textnormal{ in trial 1, and }w\textnormal{ in response to }b\textnormal{ in trial 2}\right]=f_{ab}\left(v,w\right)+g_{ab}\left(v,w\right),
\]
where 
\begin{equation}
f_{ab}\left(v,w\right)=f_{ba}\left(w,v\right),\label{eq: symmetry}
\end{equation}
and 
\begin{equation}
g_{ab}\left(v,w\right)=-g_{ab}\left(1-v,1-w\right).\label{eq: antisymmetry}
\end{equation}
It follows then that 
\[
h_{ab}\left(v,w\right)+h_{ab}\left(1-v,1-w\right)=h_{ba}\left(w,v\right)+h_{ba}\left(1-w,1-v\right),
\]
which is the QQ equation. Clearly, $F,G,H$ functions in (\ref{eq: very general v})-(\ref{eq: very general rho})-(\ref{eq: very general evolution})
can be chosen so that $f_{ab}$ and $g_{ab}$ have the desired symmetry
properties, and the QM version of $F$ and $G$ used in Ref. \cite{WANG}
(with ineffective $H$) is only one way of achieving this. It is an
open question whether one of many possible generalizations of this
QM version may turn out more profitable for dealing with opinion polling.

\paragraph{Acknowledgments. }

This work was supported by NSF grant SES-1155956 to END, AFOSR grant
FA9550-12-1-0397 to JRB, by the visiting professor fellowships for
IB and AKh within the framework of the grant of the Quantum Bio-information
Center of Tokyo University of Science, and by the Linnaeus University
Modeling of Complex Hierarchic Systems grant to IB and AKh.

\subsection*{Appendix: Proofs}

\paragraph*{\emph{Proposition \ref{prop: aa dynamics}}}

\emph{(repeated measurements) An observable $A$ has the Lüders property
if and only if $U_{\Delta}$ in (\ref{eq: unitary evolution}) is
ineffective for $A$.}
\begin{proof}
The ``if'' part is demonstrated by (\ref{eq: p' ineffiecient evolution})-(\ref{eq: p' again}).
For the ``only if'' part, let the eigenprojector $P$ correspond
to the eigenvalue $v$ of $A$. We have\textcolor{black}{
\[
p=\langle P\psi,\psi\rangle=\Vert P\psi\Vert^{2}>0,
\]
and the next state vector is 
\[
\psi^{\left(2\right)}=\frac{P\psi}{\Vert P\psi\Vert}.
\]
Following the evolution 
\[
\psi^{\left(2\right)}\rightarrow\psi_{\Delta}^{\left(2\right)}=U_{\Delta}\psi^{\left(2\right)},
\]
the L\"uders property implies
\[
p'=\langle P\psi_{\Delta}^{\left(2\right)},\psi_{\Delta}^{\left(2\right)}\rangle=1,
\]
or, equivalently,
\[
\langle PU_{\Delta}\psi^{\left(2\right)},U_{\Delta}\psi^{\left(2\right)}\rangle=\langle U_{\Delta}^{\dagger}PU_{\Delta}\psi^{\left(2\right)},\psi^{\left(2\right)}\rangle=1.
\]
}As $\Vert\psi^{\left(2\right)}\Vert=1$, and $U_{\Delta}^{\dagger}PU_{\Delta}$
is an orthogonal projection, the lengths $\left\Vert U_{\Delta}^{\dagger}PU_{\Delta}\psi^{\left(2\right)}\right\Vert $
does not exceed 1. Therefore $\langle U_{\Delta}^{\dagger}PU_{\Delta}\psi^{\left(2\right)},\psi^{\left(2\right)}\rangle=1$
implies 
\[
U_{\Delta}^{\dagger}PU_{\Delta}\psi^{\left(2\right)}=\psi^{\left(2\right)},
\]
or
\[
U_{\Delta}^{\dagger}PU_{\Delta}P\psi=P\psi.
\]
By the stability considerations (Lemma \ref{lem: stability} below,
with $X_{1}=P$ to guarantee $p>0$, and $Y=U_{\Delta}^{\dagger}PU_{\Delta},Z=P$),
\textcolor{black}{
\[
U_{\Delta}^{\dagger}PU_{\Delta}P=P.
\]
Since $P$ is Hermitian
\[
P^{\dagger}=P\left(U_{\Delta}^{\dagger}PU_{\Delta}\right)=P=\left(U_{\Delta}^{\dagger}PU_{\Delta}\right)P,
\]
so $P$ and $U_{\Delta}^{\dagger}PU_{\Delta}$ commute. Now, $U_{\Delta}^{\dagger}\left(I-P\right)U_{\Delta}=I-U_{\Delta}^{\dagger}PU_{\Delta}$,
and it commutes with $I-P$: 
\[
\left(I-U_{\Delta}^{\dagger}PU_{\Delta}\right)\left(I-P\right)=I-U_{\Delta}^{\dagger}PU_{\Delta}-P+\left(U_{\Delta}^{\dagger}PU_{\Delta}\right)P=I-U_{\Delta}^{\dagger}PU_{\Delta}-P+P\left(U_{\Delta}^{\dagger}PU_{\Delta}\right)=\left(I-P\right)\left(I-U_{\Delta}^{\dagger}PU_{\Delta}\right).
\]
Let us choose an orthonormal basis $e_{1},\ldots,e_{n}$ consisting
of the eigenvectors of $P$, so that $e_{1},\ldots,e_{k}$ are the
eigenvectors associated with eigenvalue 1 (and then the rest of the
$e$'s are the eigenvectors of $I-P$ with eigenvalue 1). In this
basis, $P$ is a diagonal matrix with the first $k$ diagonal entries
1, and the rest of them zero, and $I-P$ is a diagonal matrix with
the last $n-k$ diagonal entries 1, and the rest of them zero. We
have $Pe_{i}=e_{i}$, for $i\leq k$, and then 
\[
\left(U_{\Delta}^{\dagger}PU_{\Delta}\right)Pe_{i}=\left(U_{\Delta}^{\dagger}PU_{\Delta}\right)e_{i}=e_{i}.
\]
So, all $e_{1},\ldots,e_{k}$ are eigenvectors of $U_{\Delta}^{\dagger}PU_{\Delta}$
with eigenvalues 1. Since $U_{\Delta}^{\dagger}PU_{\Delta}$ has the
same eigenvectors $e_{1},\ldots,e_{n}$ as $P$, it is a diagonal
matrix with the first $k$ diagonal entries 1. Analogously we find
that $U_{\Delta}^{\dagger}\left(I-P\right)U_{\Delta}$ is a diagonal
matrix with the last $n-k$ diagonal entries 1. Since these matrices
add to $I$, the rest of the diagonal entries in $U_{\Delta}^{\dagger}PU_{\Delta}$
must be zero, and this means that $U_{\Delta}^{\dagger}PU_{\Delta}=P$.
By Lemma \ref{lem: inefficient}, this means that $U_{\Delta}$ and
$A$ commute.}
\end{proof}

\paragraph{Proposition \ref{prop:abab}}

\emph{(}alternating measurements\emph{). Let $A$ and $B$ possess
the Lüders property, and let the measurement sequences  
\[
\left(A,v,p_{vA}\right)\rightarrow\left(B,w,p_{wB}\right)
\]
exist for all $v,w\in\left\{ 0,1\right\} $, and some initial state
vector $\psi$. Then, in the measurement sequences
\[
\left(A,v,p_{vA}\right)\rightarrow\left(B,w,p_{wB}\right)\rightarrow\left(A,v,p'_{vA}\right),
\]
the property $p'_{vA}=1$ holds stable for this $\psi$ if and only
if $A$ and $B$ commute, $AB=BA.$} 
\begin{proof}
Let the eigenprojectors $P$ and $Q$ correspond to the eigenvalue
$0$ of $A$ and $B$, respectively. Consider the measurement sequence
\[
\left(A,0,p_{0A}\right)\rightarrow\left(B,0,p_{0B}\right)\rightarrow\left(A,0,p'_{0A}\right).
\]
Let the two inter-trial intervals be $\Delta_{1}$ and $\Delta_{2}$.
Due to the Lüders property, each of the corresponding evolution operators
$U_{1}$ and $U_{2}$ is ineffective for both $A$ and $B$ (commutes
with any of their eigenprojectors). Then the state vectors at the
beginning of each trial are
\[
\psi=\psi^{\left(1\right)}\rightarrow U_{1}\frac{P_{v}\psi}{\Vert P_{v}\psi\Vert}=\psi{}_{\Delta_{1}}^{\left(2\right)}\rightarrow U_{2}U_{1}\frac{Q_{w}P_{v}\psi}{\Vert Q_{w}U_{1}P_{v}\psi\Vert}=\psi{}_{\Delta_{2}}^{\left(3\right)}
\]
and the corresponding probabilities are 
\[
p_{0A}=\Vert P\psi\Vert^{2},\; p_{0B}=\frac{\langle QP\psi,P\psi\rangle}{\Vert P\psi\Vert^{2}},\; p_{0A}^{\prime}=\frac{\langle PQP\psi,QP\psi\rangle}{\Vert QP\psi\Vert^{2}}.
\]
The ``if'' part is proved by direct computation. If $AB=BA$, then
$PQ=QP$, and we have
\[
p_{0A}^{\prime}=\frac{\langle PQP\psi,QP\psi\rangle}{\Vert QP\psi\Vert^{2}}=\frac{\langle QP\psi,QP\psi\rangle}{\Vert QP\psi\Vert^{2}}=1.
\]
We now prove the ``only if'' part. Denoting $\phi=QP\psi/\Vert QP\psi\Vert$,
the condition 
\[
\langle P\phi,\phi\rangle=p_{0A}^{\prime}=1
\]
implies $P\phi=\phi$, because $\Vert\phi\Vert=1$, and $\Vert P\phi\Vert$
(the length of an orthogonal projection of $\phi$) does not exceed
1. Hence 
\[
PQP\psi=QP\psi.
\]
By the stability considerations (Lemma \ref{lem: stability} below,
with $X_{1}=P,X_{2}=PQP$ to guarantee $p_{0A}>0,p_{0B}>0$, and $Y=PQP,Z=QP$),
\[
PQP=QP,
\]
and, taking the conjugate transpositions,
\[
PQP=\left(PQP\right)^{\dagger}=PQ.
\]
So, $PQ=QP$. By simple algebra then, either of $P,I-P$ commutes
with either of $Q,I-Q$, and this means that $AB=BA.$
\end{proof}

\paragraph{Proposition \ref{prop: commutativity}}

(no order effect). \emph{If $A$ and $B$ possessing the L\"uders
property commute, then in the measurement sequences 
\[
\left(A,v,p_{vA}\right)\rightarrow\left(B,w,p_{wB}\right),
\]
\[
\left(B,w,q_{wB}\right)\rightarrow\left(A,v,q_{vA}\right)
\]
the joint probabilities of the two outcomes are the same, 
\[
p_{vA}p_{wB}=q_{wB}q_{vA}.
\]
Consequently, 
\[
\Pr\left[v\left(A\right)=v\textnormal{ in trial 1}\right]=\Pr\left[v\left(A\right)=v\textnormal{ in trial 2}\right]
\]
and 
\[
\Pr\left[w\left(B\right)=w\textnormal{ in trial 1}\right]=\Pr\left[w\left(B\right)=w\textnormal{ in trial 2}\right].
\]
}
\begin{proof}
For the sequence $\left(A,v,p_{vA}\right)\rightarrow\left(B,w,p_{wB}\right)$,
we have
\[
\psi=\psi^{\left(1\right)}\rightarrow U_{A}\frac{P\psi}{\Vert P\psi\Vert}=\psi{}_{\Delta_{A}}^{\left(2\right)}
\]
and the probabilities are
\[
p_{0A}=\Vert P\psi\Vert^{2},
\]
\[
p_{0B}=\langle Q\psi{}_{\Delta_{A}}^{\left(2\right)},\psi{}_{\Delta_{A}}^{\left(2\right)}\rangle=\langle QU_{A}\frac{P\psi}{\Vert P\psi\Vert},U_{A}\frac{P\psi}{\Vert P\psi\Vert}\rangle=\frac{\langle QU_{A}P\psi,U_{A}P\psi\rangle}{\Vert P\psi\Vert^{2}}.
\]
The joint probability is therefore 
\[
p_{0A}p_{0B}=\langle QU_{A}P\psi,U_{A}P\psi\rangle=\langle PQP\psi,\psi\rangle,
\]
where we have used the commutativity of the unitary operators with
the observables. Analogously, for the sequence \emph{$\left(B,w,q_{wB}\right)\rightarrow\left(A,v,q_{vA}\right)$,}
we get 
\[
p_{0B}p_{0A}=\langle QPQ\psi,\psi\rangle.
\]
But $P$ and $Q$ commute by Proposition \ref{prop:abab}, whence
\[
PQP=QPQ.
\]
This proves $p_{vA}p_{wB}=q_{wB}q_{vA}$. The other two equations
follow by presenting the probabilities in them as sums of suitably
chosen joint probabilities.
\end{proof}

\paragraph{Proposition \ref{propo: POVM}}

\emph{(no generalization)A POVM $A=\left(E_{0},E_{1}\right)$ has
the Lüders property with respect to a state $\psi$ if and only if
$A$ is a conventional observable (i.e., it is a Hermitian operator,
and its components $E_{0},E_{1}$ are its eigenprojectors).} 
\begin{proof}
The ``if'' part is obvious: if $A$ is a Hermitian operator, it
has the Lüders property with respect to any state $\psi$. 

We prove the ``only if'' part. Consider the measurement sequence
\[
\left(A,v,p_{v}\right)\rightarrow\left(A,v,p'_{v}\right),
\]
\textcolor{black}{with }$\psi=\psi^{\left(1\right)}$\textcolor{black}{.
Since 
\[
p_{v}=\langle E_{v}\psi,\psi\rangle=\langle M_{v}\psi,M_{v}\psi\rangle=\Vert M_{v}\psi\Vert^{2}>0,
\]
the next state vector (interjecting the unitary evolution operator
$U_{\Delta}$) is 
\[
\psi_{\Delta}^{\left(2\right)}=U_{\Delta}\psi^{\left(2\right)}=U_{\Delta}\frac{M_{v}\psi}{\Vert M_{v}\psi\Vert}=\frac{U_{\Delta}M_{v}\psi}{\Vert M_{v}\psi\Vert}.
\]
Since $\Vert\psi_{\Delta}^{\left(2\right)}\Vert=1$, it follows from
Lemma \ref{lemma simple} below that the equality $p'_{v}=\langle E_{v}\psi_{\Delta}^{\left(2\right)},\psi_{\Delta}^{\left(2\right)}\rangle=1$
implies $E_{v}\psi_{\Delta}^{\left(2\right)}=\psi_{\Delta}^{\left(2\right)}$,
or 
\[
E_{v}U_{\Delta}M_{v}\psi=U_{\Delta}M_{v}\psi.
\]
}By the stability considerations (Lemma \ref{lem: stability} below,
with $X_{1}=E_{v}$ to guarantee $p_{v}>0$, and $Y=E_{v}U_{\Delta}M_{v},Z=U_{\Delta}M_{v}$)\textcolor{black}{,
\[
E_{v}\left(U_{\Delta}M_{v}\right)=\left(U_{\Delta}M_{v}\right).
\]
Since 
\[
E_{v}=M_{v}^{\dagger}M_{v}=M_{v}^{\dagger}U_{\Delta}^{\dagger}U_{\Delta}M_{v}=\left(U_{\Delta}M_{v}\right)^{\dagger}\left(U_{\Delta}M_{v}\right),
\]
we can apply Lemma \ref{lemma even simpler} below (with $E_{v}=E$
and $U_{\Delta}M_{v}=S$) to establish that all eigenvalues of $E_{v}$
are 0's and 1's. Therefore $E_{v}$ is an orthogonal projector operator.
A POVM $\left(E_{0},E_{1}\right)$ with both components orthogonal
projectors is a conventional observable.}\end{proof}
\begin{lem}
\label{lem: inefficient}For a unitary operator $U$ and an observable
$A$ with eigenprojectors $P_{0},P_{1}$, if $U$ and $A$ commute
then $U^{\dagger}P_{v}U=P_{v}$ for $v=0$ and $v=1$; and if $U^{\dagger}P_{v}U=P_{v}$
for either $v=0$ or $v=1$, then $U$ and $A$ commute.\end{lem}
\begin{proof}
$U$ and $A$ commute if and only if $U$ commutes with $P_{1}$ (associated
with eigenvalue 1), because $A=P_{1}$. We should prove therefore
that
\[
UP_{1}=P_{1}U\Longleftrightarrow U^{\dagger}P_{1}U=P_{1}\Longleftrightarrow U^{\dagger}P_{0}U=P_{0}.
\]
For the first equivalence, if $UP_{1}=P_{1}U$, then $U^{\dagger}P_{1}U=U^{\dagger}UP_{1}=P_{1}$;
conversely, if $U^{\dagger}P_{1}U=P_{1}$, then $UP_{1}=UU^{\dagger}P_{1}U=P_{1}U$.
For the second equivalence, if $U^{\dagger}P_{v}U=P_{v}$, then $U^{\dagger}P_{1-v}U=U^{\dagger}\left(I-P_{v}\right)U=I-U^{\dagger}P_{v}U=I-P_{v}=P_{1-v}$.\end{proof}
\begin{lem}
\label{lem: stability}Let $X_{1},\ldots,X_{n},Y,Z$ be some matrices.
The statement
\[
\begin{array}{c}
\langle X_{1}\psi,\psi\rangle>0,\ldots,\,\langle X_{n}\psi,\psi\rangle>0\\
\Downarrow\\
Y\psi=Z\psi
\end{array}
\]
holds stable for $\psi$ if and only if $Y=Z$.\end{lem}
\begin{proof}
The ``if'' part is trivial. For the ``only if'' part, by the definition
of stability the initial state can be chosen as \textcolor{black}{
\[
\overline{\psi}=\frac{\psi+\delta}{\left\Vert \psi+\delta\right\Vert },
\]
}where $\delta\in B_{r}\left(0\right)$ (open ball of radius $r$
centered at 0). By continuity considerations, $r$ can be chosen sufficiently
small for $\langle X_{1}\psi,\psi\rangle,\ldots,\,\langle X_{n}\psi,\psi\rangle$
to remain positive. But $Y\overline{\psi}=Z\overline{\psi}$ any such
$\overline{\psi}$, and we have \textcolor{black}{
\[
Y\frac{\psi+\delta}{\left\Vert \psi+\delta\right\Vert }=Z\frac{\psi+\delta}{\left\Vert \psi+\delta\right\Vert },
\]
}whence 
\[
Y\delta=Z\delta.
\]
Since every vector is collinear to some $\delta$, $Y$ and $Z$ coincide. \end{proof}
\begin{lem}
\label{lemma simple}\textcolor{black}{Let $\left(E_{1},E_{2}\right)$
be a POVM}. If $\langle E_{v}\phi,\phi\rangle=1$ and $\Vert\phi\Vert=1$,
then $E_{v}\phi=\phi$.\end{lem}
\begin{proof}
Writing $E_{v}\phi=c\phi+\gamma$, with $\gamma\bot\phi$, we see
that $\langle E_{v}\phi,\phi\rangle=1$ implies $c=1$. Choose an
orthonormal basis $\left(e_{1},\ldots,e_{n}\right)$ in the Hilbert
space so that $\phi=e_{1}$. In this basis $\gamma=\sum_{i=2}^{n}u_{i}e_{i}$.
Assume that $\gamma\not=0$, and let, with no loss of generality,
$u_{2}\not=0$. The components $E_{v}$ in this basis is 
\[
E_{v}=\left(\begin{array}{cccc}
1 & \overline{u_{2}} & \ldots & \overline{u_{n}}\\
u_{2} & \ldots & \ldots & \ldots\\
\vdots & \vdots & \ddots & \vdots\\
u_{n} & \ldots & \ldots & \ldots
\end{array}\right),
\]
because when multiplied by $\phi=e_{1}=\left(1,0,\ldots,0\right)^{\top}$,
it should yield $e_{1}+\sum_{i=2}^{n}u_{i}e_{i}=\phi+\gamma$. Then
the other component in this basis is 
\[
E_{1-v}=\left(\begin{array}{cccc}
0 & -\overline{u_{2}} & \ldots & -\overline{u_{n}}\\
-u_{2} & \ldots & \ldots & \ldots\\
\vdots & \vdots & \ddots & \vdots\\
-u'_{n} & \ldots & \ldots & \ldots
\end{array}\right),
\]
because $E_{v}+E_{1-v}=I$. But with $u_{2}\not=0$, the leading principal
minor 
\[
\left|\begin{array}{cc}
0 & -\overline{u{}_{2}}\\
-u{}_{2} & \ldots
\end{array}\right|<0,
\]
which contradicts the requirement that $E_{1-v}$ be positive semidefinite.
This contradiction shows that $\gamma=0$.\end{proof}
\begin{lem}
\textcolor{black}{\label{lemma even simpler}Let $E=S^{\dagger}S$
be a component of a POVM, and let $ES=S$. Then all eigenvalues of
$E$ are 0's and 1's.}\end{lem}
\begin{proof}
\textcolor{black}{Since $E$ is Hermitian, we can select a basis consisting
of its eigenvectors. In this basis matrix $E$ is diagonal with the
diagonal elements $\lambda_{1},...,\lambda_{n}$ (the eigenvalues
of $E$). Suppose that one of these elements, say $\lambda_{1}$,
is not 1. From $ES=S$ it follows that $\lambda_{1}s_{1}=s_{1}$,
where $s_{1}$ is the first row of $S$. Therefore $s_{1}$ of $S$
consists of zeros. But since $E=S^{\dagger}S$, we have $\lambda_{1}=\langle s_{1},s_{1}\rangle=0$.
This proves that $\lambda_{1},...,\lambda_{n}$ consists of 0's and
1's.}\end{proof}

\end{document}